\documentclass[a4paper,10pt,english]{article}

\usepackage[margin=1.25in]{geometry}

\usepackage{cite}    

\usepackage[latin1]{inputenc}

\usepackage{amsthm}
\usepackage{amsmath}
\usepackage{amssymb}

\newtheorem{theorem}{Theorem}
\newtheorem{defn}[theorem]{Definition}

\newtheorem{lemma}[theorem]{Lemma}
\newtheorem{prop}[theorem]{Proposition}

\newcommand{\eps}{\varepsilon}
\newcommand{\e}{1+\eps}

\newcommand{\abc}[3]{\ensuremath{\textup{#1}|\,#2\,|#3}}
\newcommand{\I}{\mathcal{I}}

\newcommand{\C}{\mathcal{C}}

\newcommand{\opt}{\textup{\sc Opt}}
\newcommand{\OPT}{\opt}
\newcommand{\MS}{\textup{\sc MS}}
\newcommand{\alg}{\textup{\sc Alg}}

\usepackage{babel}
\usepackage{paralist}

\begin{document}
\global\long\def\oe{1+O(\epsilon)}

\title{Competitive-Ratio Approximation Schemes for Minimizing the
  Makespan 
in the Online-List Model}

\author{Nicole Megow\thanks{Department of Mathematics, Technische Universit{\"a}t
       Berlin, Germany.  Email: \texttt{nmegow@math.tu-berlin.de}. Supported by the German Science Foundation (DFG) under contract  ME 3825/1.}%
     \and Andreas Wiese\thanks{Max-Planck-Institut für Informatik,
       Saarbrücken, Germany. Email: \texttt{wiese@mpi-inf.mpg.de}.}
}
\date{}

\maketitle

\begin{abstract}
  We consider online scheduling on multiple machines for jobs arriving
  one-by-one with the objective of minimizing the  makespan. For any
  number of identical parallel or uniformly related machines, we provide a competitive-ratio approximation
  scheme that computes an online algorithm whose competitive ratio is arbitrarily
  close to the best possible competitive ratio. We also determine this value
  up to any desired accuracy. This is the first application of competitive-ratio approximation
  schemes in the online-list model. The result proves the
  applicability of the concept in different online models. We expect that it fosters further
  research on other online problems. 
\end{abstract}

\section{Introduction}

Online scheduling problems have been studied extensively for more than
two decades~\cite{sgall98,pruhs04}. One of the most extensively
investigated problems among them is the makespan minimization
problem with jobs arriving one-by-one: We are given~$m$ identical parallel machines, and we assume
throughout the paper that~$m$ is an arbitrary but fixed constant. The set of
jobs~$J=\{1,2,\ldots\}$ with integral processing times~$p_j\geq 1$~($j\in
J$) is presented to the online algorithm one after the other. Once a
job is present, it must be assigned without splitting, immediately, and irrevocably to a
machine before the next job is revealed. This model for revealing
online information one-by-one is called  the {\em online-list}
model~\cite{pruhs04}. The goal is to minimize the makespan, that is, the last
completion time of all currently present jobs. Using the standard
three-field notation~\cite{grahamLLR79}, we denote the problem as the online-list variant
of~\abc{Pm}{\!\!}{C_{\max}}. We also consider the more general model of
uniformly related machines~\abc{Qm}{\!\!}{C_{\max}}, where each
machine~$i\in\{1,\ldots,m\}$ is given a speed~$s_i$ and the execution
time of job~$j$ on machine~$i$ is~$p_j/s_i$. 

The performance of online algorithms is typically assessed by {\em
  competitive analysis}~\cite{sleatorT85,karlinMRS88} which
determines the worst-case performance compared to an optimal offline
algorithm. Let an instance~$I$ be defined by a set of jobs~$J$ with
processing times~$p_j$~($j\in J$), and~$m$, the number of available
machines. And let~$\I_m$ be the set of instances with~$m$
machines. We call an online algorithm $\rho(m)$-{\em competitive}
if, for any problem
instance~$I\in \I_m$, it achieves a solution with cost~$\alg(I) \leq \rho(m)
\cdot \opt(I)$, where~$\alg(I)$ and $\opt(I)$ denote the solution
value of the online and an optimal offline algorithm, respectively,
for the same instance~$I$. The \emph{competitive ratio}~$\rho_{\alg}(m)$ of~$\alg$ is the infimum
over all~$\rho$ such that~$\alg$ is~$\rho$-competitive. The minimum
competitive ratio~$\rho^*(m)$ achievable by any online algorithm for
instances in $\I_m$ is called {\em optimal for~$m$ machines}. The
  optimal competitive ratio over all number of machines
  is~$\rho^*:=\max_{m\in \mathbb{N}} \rho^*(m)$.

Only recently, the concept of competitive-ratio approximation schemes
was introduced in~\cite{guentherMMW13}. Such an approximation scheme
is a procedure that computes a nearly optimal online algorithm and at the same time provides a
nearly exact estimate of the optimal competitive ratio. 
The general definition (without distinguishing by~$m$) is as follows.
\begin{defn}
  A {\em competitive-ratio approximation scheme} computes for a
  given~$\eps>0$ an online algorithm~$\alg$ with a competitive ratio~$\rho_{\alg}
  \leq (1+\eps)\rho^*$. Moreover, it determines a value~$\rho'$ such
  that $\rho'\leq \rho^* \leq (1+\eps) \rho'$.
\end{defn}

In this paper we provide a competitive-ratio approximation scheme for
the online-list variant of makespan minimization on identical parallel
and uniformly related machines for any number of machines. This is the first competitive-ratio approximation
scheme for a problem in the {\em online-list} model in contrast to
previous work in the so-called  {\em online-time}
model~\cite{pruhs04}. In the latter model jobs are revealed to the algorithm
online over time at their individual release date. Regarding the
decision making process, an online algorithm has more freedom in this model, as it is
allowed to postpone decisions or even revoke them as long as the jobs have
not been executed. 

\subsection{Related Work}

The online-list makespan minimization
problem has been studied extensively---mainly, on identical parallel machines. The classical list scheduling
algorithm with competitive ratio~$2-1/m$~\cite{graham66} is optimal
for~$m \in {2,3}$~\cite{faigleKT89}. For~$m\geq 4$ better algorithms
have been proposed and improved general lower bounds were shown in a
series of
works~\cite{faigleKT89,FleischerWahl2000,RudinChandrasekaran2003,GalambosW93,ChandrasekaranCGNV97,BartalFKV95,KargerPT96,BartalKR94,Albers99}. The
currently best known bounds on the optimal competitive ratio~$\rho^*(m)$  for some
particular values of~$m$ are~$\rho^*(4)\in[1.732,1.733]$,
$\rho^*(5)\in[1.770,1.746]$, $\rho^*(6)\in[1.8, 1.773]$, \ldots,
$\rho^*\in [1.88,1.9201]$.


On uniformly related machines the gap is much larger. The lower bound
on the competitive ratio for an arbitrary number of machines
is~$2.564$~\cite{ebenlendrS12}, while the currently best known upper bound
is~$5.828$~\cite{bermanCK00}. Interestingly, the special case of two
related machines is completely solved, meaning that the exact
competitive ratio known~\cite{choS80} and even stronger, the exact
ratio for any pair of speeds is known~\cite{epsteinNSSW01,wenD98}.

We also remark that the preemptive variant of the identical machine problem is completely
solved and the optimal competitive ratio for any number of
machines~\cite{chenVW95} is known. For uniformly related machines, an
optimal online algorithm is known for any number of machines and any combination of
speeds~\cite{ebenlendrJS09}. Interestingly, from their linear
programming based approach it is not clear how to derive the actual value of the
competitive ratio except for~$m\in\{3,4\}$. 

Competitive-ratio approximation schemes were introduced by
G\"unther et al.~in~\cite{guentherMMW13}. They focussed on scheduling
problems in the online-time model and provide 
such schemes for various scheduling problems~\abc{Pm}{r_j,
  (pmtn)}{\sum w_jC_j}, \abc{Qm}{r_j, (pmtn)}{\sum w_jC_j}~(assuming a constant range of machine speeds
without preemption), and \abc{Rm}{r_j,pmtn}{\sum w_jC_j}. They also
consider minimizing the makespan, $C_{\max}$, and~$\sum_{j\in
  J}w_jf(C_j)$, where~$f$ is an arbitrary monomial function with fixed
exponent. Subsequently, Kurpisz et al.~\cite{kurpiszMG12} showed how to
construct competitive-ratio approximation schemes for
\abc{Rm}{r_j}{C_{\max}}, makespan minimization in a job shop
problem, and scheduling with delivery times---again, all in the online-time model. 

We are not aware of any publication of similar results for the
online-list model\footnote{However, as of writing this we got contacted by another
group of researchers~\cite{chenYZ13} who obtained similar results as
ours, independently.}. Notice that the results in~\cite{ebenlendrJS09}
are conceptually strongly related. The main difference is that our
approximation scheme provides the algorithmic means to compute the actual value of the optimal
competitive ratio (up to some error), whereas this remains open for the
algorithm in~\cite{ebenlendrJS09} even though it is provably optimal. 

\subsection{Our Results}

In this paper we provide competitive-ratio approximation schemes for
the online-list variants of makespan minimization on identical parallel
and uniformly related machines for any number of machines, that
is,~\abc{Pm}{\!\!}{C_{\max}} and~\abc{Qm}{\!\!}{C_{\max}}. 
More precisely, given~$\eps>0$ and an~$m\in \mathbb{N}$, we provide an
online algorithm~$\alg(m)$ with a competitive ratio~$\rho_{\alg}(m)
\leq (1+\eps)\rho^*(m)$. Moreover, it determines a value~$\rho'$ such
that $\rho'\leq \rho^*(m) \leq (1+\eps) \rho'$.

On a high level, we use a similar approach as in~\cite{guentherMMW13}.
We first simplify and structure the input without changing the
instance too much, and then we reduce the
complexity of possible online algorithms by interpreting them as an
algorithm map that bases its decisions only on the currently
unfinished jobs and the schedule history. The key insight is that a
very limited (constant) part of the schedule history is sufficient to
take decisions which are close to optimal. In contrast to the previous
work in the online-time model, our amount of history that has to be
considered depends on the size of the currently largest revealed job instead of
the time at which the last jobs where released.

The main purpose of this paper is to provide a proof of concept for competitive-ratio
approximation schemes and to show that it is also applicable to online
problems in the online-list model. The actual gaps between upper and lower bounds on the optimal competitive ratios
are rather small on identical parallel machines, but this is one of the most classical online-list
problems. Since many online problems follow the online-list paradigm,
we hope that this work fosters further research on competitive-ratio
approximation schemes.

\paragraph{Outline of the paper.} We first consider identical parallel
machines in Section~\ref{sec:identical}. Then we argue on how to extend
this technique to uniformly related machines in
Section~\ref{sec:uniform}. We conclude with open questions and
further research potential.

\section{Identical Parallel Machines}
\label{sec:identical}

We give a competitive-ratio approximation scheme for the online-list
variant of~\abc{Pm}{\!\!}{C_{\max}}. We first give some transformations that simplify the
input and reduce the structural complexity of online schedules. Then
we use the an abstract view on online algorithms to reduce complexity
further and to describe the approximation scheme. 

\subsection{Restrictions at $\e$ loss}

We will use the terminology that {\em at~$\e$ loss we can restrict} to instances or schedules with certain properties.
This means that we lose at most a factor~$\e$ in the objective value, as~$\eps\rightarrow 0$, 
by limiting our attention to those instances. We bound several relevant
parameters by constants. If not stated differently, any mentioned
constant depends only on~$\eps$ and~$m$. 

In the online-list model we refer to an iteration for each job
arrival. We will slightly abuse notation and refer to the iteration in
which job~$j$ is revealed as iteration~$j$. For an online algorithm
$A$, an instance $I$, and an iteration $j$, denote by $A_{j}(I)$ the
makespan of the schedule obtained after iteration $j$ when processing
instance $I$ by algorithm $A$. Furthermore, we define~$p(J):=\sum_{j\in J} p_j$.

The first observation has been made already in other contexts, e.g.,
in~\cite{afrati99} for minimizing~$\sum w_jC_j$.
\begin{prop}
  \label{lem:p_j-power-1+eps}At $\e$ loss we can restrict to instances
  where all $p_{j}$ are powers of $\e$.
\end{prop}

In order to simplify the construction of our algorithms, we can actually
restrict to instances with a very simple and special structure.

\begin{lemma}
  \label{lem:few-jobs-each-processing-time}
  At $\e$ loss we can restrict to instances where for each $k\in\mathbb{N}$ there are at most $\frac{m}{\epsilon^{3}}$
  jobs $j$ with $p_{j}=(\e)^{k}$.
\end{lemma}
\begin{proof}
  Suppose that we have an online algorithm $A'$ which achieves a
  competitive ratio of $\rho_{A'}(m)$ on instances with at most $\frac{m}{\epsilon^{3}}$
jobs, each with processing time~$p_{j}=(\e)^{k}$ for each
$k\in\mathbb{N}$. Based on $A'$ we construct an online algorithm $A$ for arbitrary instances
(assuming that processing times are powers of $\e$) with a competitive
ratio of $\rho_{A}(m)\le(\oe)\rho_{A'}(m)$. 

Suppose that we are given an~(arbitrary) instance $I$ where all $p_{j}$ are powers of $\e$.
We construct an instance $I'$ which we present to $A'$. Based on the
schedule $A'(I')$ we construct the schedule $A(I)$. As long as for
each $k\in\mathbb{N}$ at most $\frac{m}{\epsilon^{3}}$ jobs $j$
with $p_{j}=(\e)^{k}$ are released, the instances $I$ and $I'$
are identical and we define $A(I)$ to be identical to $A'(I')$.
Now suppose that in some iteration $j$ a job $j$ with $p_{j}=(\e)^{k}$
is revealed after there where released already $\frac{m}{\epsilon^{3}}$ other
jobs with the same processing time. Let~$I'_j$ denote the instance up
to job~$j$. Observe that
$\opt(I'_{j})\ge\frac{1}{\epsilon^{3}}(\e)^{k}$. Let
$p:=(\e)^{\left\lceil \log_{\e}\epsilon^{2}\cdot
    \opt(I'_{j})\right\rceil } \geq \eps^2 \cdot
    \opt(I'_{j})$.

We observe that so far at most $\frac{m}{\epsilon^{2}}<\frac{m}{\epsilon^{3}}$
jobs of size $p$ have been released since otherwise $\OPT(I'_{j})\ge\frac{1}{m}(\frac{m}{\epsilon^{2}}+1)\cdot p>\frac{1}{m}\cdot\frac{m}{\epsilon^{2}}\cdot\epsilon^{2}\OPT(I'_{j})=\OPT(I'_{j})$.
Instead of $j$, in instance $I'$ we release a new job $j'$ with
$p_{j'}=p$. Suppose that algorithm $A'$ assigns $j'$ on machine
$i$. Then, algorithm $A$ assigns the next upcoming jobs $p_{j''}$
with $p_{j''}\le p_{j}$ to $i$, as long as their total processing
time is bounded by $p$. More precisely, we define $j_{\max}$ to
be the maximum value such that for the set $J:=\{j''\in I|j\le j''\le j_{\max}\wedge p_{j''}\le p_{j}\}$
it holds that $p(J)\le p$. We define that algorithm $A$ assigns
all jobs in $J$ to machine $i$. We call $j'$ a \emph{container
job}. We say that after iteration $j_{\max}$ the container job $j'$
is \emph{full}. Intuitively this means that we do not add any further
jobs to $j'$. At each iteration $j''$ we say that the jobs in $J\cap J_{j''}$
are in the container $j'$. Observe that $p_{j'}\le p_{j}=(\e)^{k}\le\epsilon^{3}\cdot \OPT(I'_{j})\le\epsilon\cdot p$
for all $j'\in J$.

By construction, we observe that for each $k\in\mathbb{N}$ there
is at most one container job of size $(\e)^{k}$ which is less than
a $(1-\epsilon)$-fraction full.  In particular, if we create a new
container job of size $p:=(\e)^{\left\lceil \log_{\e}\epsilon^{2}\cdot \OPT(I_{j}')\right\rceil }$
then up to iteration $j$ strictly less than $\frac{m}{\epsilon^{3}}$
jobs (container jobs \emph{and }normal jobs!) of size $p$ have been
released since otherwise $\OPT(I_{j}')>\frac{1}{m}(\frac{m}{\epsilon^{3}}\cdot p)\ge\frac{1}{m}\cdot\frac{m}{\epsilon^{3}}\cdot\epsilon^{2}\OPT(I_{j}')\ge \OPT(I_{j}')$.
According to the above definition, it can happen that we open a new
container job while an old container with smaller size is not yet
to a $(1-\epsilon)$-fraction full. In this case we close the smaller
old container and do not add any further jobs to it.

To prove the competitive ratio of $A'$ we need to show that
$\frac{A(I_{j})}{\OPT(I_{j})}\le(\oe)\frac{A'(I'_{j})}{\OPT(I'_{j})}\le(\oe)\rho_{A'}(m)$. By
construction we have that $A(I_{j})\le A'(I'_{j})$ since $A$ and $A'$ assign the jobs in $J(I_{j})\cap J(I'_{j})$ to the same
machines and on each machine $i$ the total processing time of the
jobs in $J(I_{j})\setminus J(I'_{j})$ is bounded from above by the
total processing time of the jobs in $J(I'_{j})\setminus J(I{}_{j})$
(the container jobs) on this machine. It remains to show that $\OPT(I'_{j})\le(\oe)\OPT(I{}_{j})$.
Based on $\OPT(I_{j})$ we construct a schedule $S$ for $I'_{j}$
whose makespan is bounded by $(\oe)\OPT(I{}_{j})$. In $S$, we assign
all jobs in $J(I_{j})\cap J(I'_{j})$ to the same machine as in $\OPT(I_{j})$.
Then, we assign the jobs in $J(I'_{j})\setminus J(I{}_{j})$ (the
container jobs) greedily. If after the greedy assignment the global
makespan does not change, then $\OPT(I'_{j})\le \OPT(I_{j})$. 

Now suppose that after the greedy assignment the global makespan increases.
Then the load of any two machines can differ by at most $\tilde{p}$
which denotes the maximum processing time of a container job in $I'_{j}$.
Note that $\tilde{p}\le(\e)^{\left\lceil \log_{\e}\epsilon^{2}\cdot \OPT(I_{j}')\right\rceil }$
and observe that the makespan of $S$ is upper-bounded by $\frac{1}{m}\cdot p(I'_{j})+\tilde{p}$.
Since for each $k\in\mathbb{N}$ there is at most one container job
of size $(\e)^{k}$ which is less than a $(1-\epsilon)$-fraction
full we further conclude that

\begin{eqnarray*}
\OPT(I'_{j}) & \le & \frac{1}{m}\cdot p(I'_{j})+\tilde{p}\\
 & \le & \underset{\mathrm{normal\, jobs}\,\mathrm{and\,(1-\epsilon)-full\, container\, jobs}}{\underbrace{\frac{1}{m}(1+O(\epsilon))p(I{}_{j})}}+\tilde{p}+\underset{\mathrm{less\, than\,}(1-\epsilon)\mathrm{-full\, container\, jobs}}{\underbrace{\sum_{1\le k'\le\left\lceil \log_{\e}\epsilon^{2}\cdot \OPT(I_{j}')\right\rceil }(\e)^{k'}}}\\
 & \le & \frac{1}{m}(1+O(\epsilon))p(I{}_{j})+2(\e)^{\left\lceil \log_{\e}\epsilon^{2}\cdot \OPT(I_{j}')\right\rceil }+\sum_{1\le k'\le\log_{\e}\epsilon^{2}\cdot \OPT(I_{j}')}(\e)^{k'}\\
 & \le & \frac{1}{m}(1+O(\epsilon))p(I{}_{j})+2(\e)^{\left\lceil \log_{\e}\epsilon^{2}\cdot \OPT(I_{j}')\right\rceil }+\frac{1}{\epsilon}(\e)^{1+\log_{\e}(\epsilon^{2}\cdot \OPT(I_{j}'))}\\
 & \le & \frac{1}{m}(1+O(\epsilon))p(I{}_{j})+2(1+\epsilon)\epsilon^{2}\cdot \OPT(I_{j}')+\frac{1+\epsilon}{\epsilon}(\epsilon^{2}\cdot \OPT(I_{j}'))\\
 & \le & (1+O(\epsilon))\OPT(I_{j})+O(\epsilon)\cdot \OPT(I_{j}')
\end{eqnarray*}
which implies that $\OPT(I'_{j})\le(\oe)\OPT(I{}_{j})$.
\end{proof}

\subsection{Online Algorithms and Algorithm Maps}

As in~\cite{guentherMMW13} we use an abstract characterization of an online
algorithm and interpret it as a map. The map gets as input the so far computed schedule
and the size of the next released job $j$. Based on these data, it
decides to which machine it assigns $j$. 

To this end, we define a configuration $C$ as follows.
\begin{defn}
  \label{def:conf}
  A configuration $C$ is the combination of
  \begin{compactitem}
  \item a set $J(C)$ of previously released jobs, including their order,
  \item a map $\chi_{C}:J(C)\rightarrow\{1,...,m\}$ which defines the assignment
    of the jobs in $J(C)$ to the machines,
  \item the processing time $p_{j^{*}}$ of the newly released but not yet
    assigned job.
  \end{compactitem}
\end{defn}
We will write only $J$ or $\chi$ (rather than $J(C)$ or $\chi_C$) when $C$ is clear from the context.
Let~$\C$ denote the (infinite) set of configurations. We say that
a configuration $C$ is in \emph{phase $k$} if $\max_{j\in J(C)\cup\{j^{*}\}}p_{j}=(\e)^{k}$.
Let~$C$ be a configuration in phase $k$. We call a job $j\in J(C)$
\emph{relevant} if $p_{j}\ge(\e)^{k-s}$ where $s\in\mathbb{N}$ is
the smallest integer such that $s\ge\log_{\e}\frac{m(\e)}{\epsilon^{5}}$
(note that $s$ depends only on $\epsilon$ and $m$ and is independent
of $C$). Denote by $J_{R}(C)\subseteq J(C)\cup\{j^{*}\}$ all {\em relevant}
jobs for a configuration~$C$. It will turn out later that at $\e$
loss we can neglect the jobs which are not relevant, which we will
call \emph{irrelevant}. Define by $\MS(C):=\max_{i}p(\chi_{C}^{-1}(i))$
the \emph{makespan of $C$. }Also, we define $\OPT(C):=\OPT(J(C))$
(note here that $J(C)$ does not include the newly released job $j^{*}$).

We interpret an online algorithm for our problem on $m$ machines
as a map $f:\C\rightarrow\{1,...,m\}$: Given a configuration $C$
with a newly released job $j^{*}$, the algorithm map $f$ assigns
$j^{*}$ to the machine $f(C)$. Like for online algorithms, we denote
by $\rho_{f}(m)$ the competitive ratio obtained by the map $f$. 
\begin{prop}
  For each online algorithm $A$ for the problem on~$m$ machines there is an algorithm map $f$ such
  that $\rho_{f}(m)=\rho_{A}(m)$.
\end{prop}
\begin{defn}
  \label{def:conf-equiv}Let $C,C'$ be two configurations which are
  in phases $k$ and $k'$, respectively. They are \emph{equivalent}
  if there is a bijection $\sigma:J_{R}(C)\rightarrow J_{R}(C')$ such
  that
  \begin{compactitem}
  \item $p_{\sigma(j)}=(\e)^{k'-k}\cdot p_{j}$, 
  \item $\chi_{C}(j)=\chi_{C'}(\sigma(j))$ for all $j\in J_{R}(C)\setminus\{j^{*}\}$,
    and 
  \item $\sigma(j^{*})=j'^{*}$ if $j^{*}\in J_{R}(C)$, where $j'^{*}$ denotes
    the newly released job in $C'$.
  \end{compactitem}
\end{defn}

\begin{prop}
  \label{prop:O(1)-equiv-classes}There are only constantly many equivalence
  classes of configurations.
\end{prop}

In Definition~\ref{def:conf-equiv} we neglect the jobs which are
not relevant. This is justified by the following lemma.

\begin{lemma}
\label{lem:neglect-small-jobs}Let $C$ be a configuration for a phase
$k$. Then $p(J(C)\setminus
J_{R}(C))\le\epsilon\cdot(\e)^{k}\le\epsilon\cdot
\OPT(C)\le\epsilon\cdot \MS(C)$.
\end{lemma}
\begin{proof}
Recall that by Lemma~\ref{lem:few-jobs-each-processing-time}  we assumed that for each $k'\in\mathbb{N}$ there are
at most $\frac{m}{\epsilon^{3}}$ jobs $j$ with $p_{j}=(\e)^{k'}$. Hence, the
total processing time of irrelevant jobs in $J(C)$ is bounded by 
\begin{align*}
p(J(C)\setminus J_{R}(C)) & \le \sum_{1\le k'\le
  k-s}\frac{m}{\epsilon^{3}}(\e)^{k'} 
  \le  \frac{m}{\epsilon^{3}}\cdot\frac{(\e)^{k-s+1}}{\epsilon} 
  =  (\e)^{k}\frac{m(\e)^{1-s}}{\epsilon^{4}}\\
 & \le  \epsilon(\e)^{k}.
\end{align*}
where the last inequality follows since $\frac{m(\e)^{1-s}}{\epsilon^{4}}\le\epsilon$
by definition of $s$. Since~$C$ is a configuration in phase~$k$, and
thus, by definition 
a job $j$ with $p_{j}=(\e)^{k}$ must have been released. It follows that
$\epsilon\cdot(\e)^{k}\le\epsilon\cdot \OPT(C)\le\epsilon\cdot
\MS(C)$. 
\end{proof}
In particular, two equivalent configurations have almost the same
makespan and their respective jobs have almost the same optimal makespan.

\begin{lemma}
  Let $C,C'$ be two equivalent configurations for phases $k$ and $k'$,
  respectively. Then $\MS(C')\le(\oe)(\e)^{k'-k}\cdot \MS(C)$ and
  $\OPT(C')\le(\oe)(\e)^{k'-k}\cdot \OPT(C)$.
\end{lemma}
\begin{proof}
  Follows from Definition~\ref{def:conf-equiv} and
  Lemma~\ref{lem:neglect-small-jobs}.
\end{proof}

\begin{lemma}
  At $\e$ loss we can restrict to instances $I$ such that for each
  iteration $j$ we have that $p_{j}\ge\max_{j'<j}p_{j'}\cdot(\e)^{-s}$.
\end{lemma}
\begin{proof}
When in some iteration $j$ a job $j$ is released with $p_{j}<\max_{j'<j}p_{j'}\cdot(\e)^{-s}$
we assign $j$ to some arbitrary machine. By Lemma~\ref{lem:neglect-small-jobs}
the total processing time of such jobs, released up to some iteration
$j'$, is bounded by $\epsilon\cdot \OPT_{j'}$.
\end{proof}

\begin{lemma}
\label{lem:algo-maps-equiv}At $\e$ loss we can restrict to algorithm
maps $f$ such that $f(C)=f(C')$ for any two equivalent configurations
$C$ and $C'$.\end{lemma}
\begin{proof}
Let $f$ be an algorithm map without this property with a competitive
ratio of $\rho_{f}(m)$ on~$m$ machines. Based on $f$ we construct a new algorithm map
$g$ \emph{with }the claimed property such that $\rho_{g}(m)\le(\e)\rho_{f}(m)$.

We call a configuration $C$ \emph{realistic for $f$}, if there is
an instance $I$ such that $f$ ends in configuration $C$ when processing
$I$. For each equivalence class $\C'$ of configurations containing
at least one realistic configuration, we pick a realistic representant
$C\in\C'$. We say that $C$ \emph{represents} $\C'$. For all configurations
$C'\in\C'$ equivalent to $C$ we define $g$ such that $g(C')=f(C)$. 

We claim that $g$ is always in a configuration $C$ such that there
is a configuration $C'$ with $C\sim C'$ such that $C'$ is realistic
for $f$. We prove the claim by induction over the iterations. We
start with the base case of zero previous iterations. Let $C$ be
a configuration which is realistic for $g$. Since so far no jobs
have been scheduled $C$ is also realistic for $f$. Now suppose that
the claim is true when the first $\ell$ jobs have been released.
Suppose that after $\ell$ jobs have been released $g$ is in a configuration
$C$ for a phase $k$ such that some configuration $C'$ for phase
$k'$ with $C\sim C'$ is realistic for $f$. Assume w.l.o.g.~that
$C'$ represents its equivalence class. Assume that in $C$ a job
$j^{*}$ with processing time $p_{j^{*}}$ is released and denote
by $j'^{*}$ the newly released job in $C'$.

By construction, $g$ assigns $j$ to the machine $g(C)=f(C')$. Then
a new (relevant) job $\bar{j}$ is released which yields the configuration
$\bar{C}$. Denote by $\bar{C}'$ the configuration which results
from $C'$ after assigning job $j'$ to machine $f(C')$ and the release
of a new job $\bar{j}'$ with $p_{\bar{j}'}=(\e)^{k'-k}p_{\bar{j}}$.
Since $\bar{j}$ is relevant for $\bar{C}$, it follows that $\bar{j}'$
is relevant for $\bar{C}'$. Since $C\sim C'$ and $p_{\bar{j}'}=(\e)^{k'-k}p_{\bar{j}}$
we conclude that $\bar{C}\sim\bar{C}'$.
\end{proof}

The decision of an algorithm map (with all above simplifications)
for a configuration $C$ depends only on the equivalence class of
$C$. Since there are only constantly many equivalence classes for
configurations (see Proposition~\ref{prop:O(1)-equiv-classes}) and
for each configuration there are only $m$ possible decisions, there
are only constantly many algorithm maps. Hence, we can enumerate them
all. With the procedure given by the following lemma we estimate its
competitive ratio. Finally, we output the map with the minimum estimated
competitive ratio.

\begin{lemma}
  Let $f$ be an algorithm map for~$m$ machines. There exists an algorithm which computes
  a value $\bar{\rho}$ with $\rho_{f}(m) \le\bar{\rho} \le(\e)\rho_{f}(m)$.
\end{lemma}
\begin{proof}
  In order to determine $\rho_{f}(m)$ it is sufficient to know all possible
  realistic configurations for $f$. By
  Lemma~\ref{lem:few-jobs-each-processing-time}, the realistic configuration $C$
  with the worst competitive ratio determines $\rho_{f}(m)$ up to an error
  of $\e$. 
\end{proof}

Combining all statements gives our main theorem.

\begin{theorem}
  There is a competitive-ratio approximation scheme for the online-list variant
 of the problem \abc{Pm}{\!\!}{C_{\max}} for any number of machines~$m$.
\end{theorem}

\section{Uniformly Related Machines}
\label{sec:uniform}

With small additional instance transformations we can apply a similar
construction in the setting of uniformly related machines. By scaling processing times
and machine speeds, we can assume w.l.o.g.~that the slowest
machine has unit speed. Let~$s_{\max}$ denote the speed of the fastest
machine in a given instance.


\begin{prop}
  At $\e$ loss we can assume that the speed of each machine is a power
  of $\e$.
\end{prop}

\begin{lemma}
At $1+\eps$ loss, we can restrict to instances in which
$s_{\max}$ is bounded by $m/\epsilon$.
\end{lemma}
\begin{proof}
Take a given schedule with makespan~$\MS$ on related machines with speed values
$s_{1},...,s_{\max}$. For each
machine whose speed is at most $\frac{\epsilon}{m}\cdot s_{\max}$,
we take the jobs assigned to it and add them to the
fastest machine. The moved processing volume increases the total
processing volume on the fast machine by at most~$\frac{\epsilon}{m}
\cdot \MS$. Thus, we can simply ignore machines whose speed is
at most $\frac{\epsilon}{m}\cdot s_{\max}$. The remaining machines
have speeds in the range of $[\frac{\epsilon}{m}\cdot
s_{\max},s_{\max}]$. Since we assume that the slowest machine has unit speed, after rounding
the speeds we have that $s_{\max}\le m/\epsilon$. 
\end{proof}

Hence, for each value $m$ there are only finitely many speed vectors
$s_{1},...,s_{m}$. For each of these speed vectors, we can bound the number
of jobs of the same size similarly to
Lemma~\ref{lem:few-jobs-each-processing-time} with an additional
dependence on~$s_{\max}\leq m/\eps$.
This allows us to define configurations, algorithm maps, and
equivalence relations similarly as in the previous section.

\begin{theorem}
  There is a competitive-ratio approximation scheme for the online-list variant
 of the problem \abc{Qm}{\!\!}{C_{\max}} for any number of machines~$m$.
\end{theorem}

\section{Conclusion and Further Research}

We provide competitive-ratio approximation schemes for the makespan
minimization problem when jobs arrive online over a list. This proves
that the concept of competitive-ratio approximation schemes is not
limited to online (scheduling) problems in the online-time model. 

The approximation schemes presented in this paper, compute a nearly
optimal solution for any number of machines. On the theoretical side, it would be interesting
to give a general approximation of the optimal competitive ratio over
all possible numbers of machines $m$. This requires a better understanding of how~$\rho^*(m)$
behaves as a function of~$m$. We know that it is bounded from
above~(for every~$m$). It seems intuitively (and suggested by
known bounds) true that it is increasing in~$m$.


Our approximation schemes do not only determine nearly best possible
online algorithms, they also provide the algorithmic tools to compute
the value of the optimal competitive ratio up to any desired accuracy. This
is interesting because it contrasts the common approach to derive upper
and lower bounds on the (optimal) competitive ratio manually. In
particular, our theory proves that a computer may execute the
algorithm to compute the desired bounds. However, the drawback of our
presented construction is its computational complexity. To reduce the
gaps between the currently best known upper and lower bounds, we would
have to chose a quite small accuracy parameter~$\eps$ which leads to a
hopeless running time. We believe that a more careful design of the
necessary input simplification and algorithm structuring might lead to
approximation schemes that can compute explicitly the value of
improved bounds. 

The current research on competitive-ratio approximation schemes
focussed on particular online scheduling problems. Our vision is to
use insights for particular problems to eventually characterize
general properties of online problems that allow for a
competitive-ratio approximation scheme.





\end{document}